\documentclass[10pt, amsmath, amssymb]{amsart}

\usepackage{amsmath}
\usepackage{amsthm}
\usepackage{amssymb}
\usepackage{graphics}
\usepackage{bm}

\newcommand{\be}{\begin{equation}}
\newcommand{\ee}{\end{equation}}
\newcommand{\ba}{\begin{eqnarray}}
\newcommand{\ea}{\end{eqnarray}}
\newcommand{\bi}{\begin{itemize}}
\newcommand{\ei}{\end{itemize}}
\newcommand{\bn}{\begin{enumerate}}
\newcommand{\en}{\end{enumerate}}
\newcommand{\bp}{\begin{proof}}
\newcommand{\ep}{\end{proof}}

\newcommand{\mc}{\ensuremath{\mathcal}}
\newcommand{\mf}{\ensuremath{\mathfrak}}
\newcommand{\ov}{\ensuremath{\overline}}

\newcommand{\La}{\ensuremath{\Lambda }}
\newcommand{\la}{\ensuremath{\lambda }}
\newcommand{\ka}{\ensuremath{\kappa }}

\newcommand{\eps}{\ensuremath{\epsilon }}
\renewcommand{\bm}{\ensuremath{\mathbb }}

\newcommand{\ip}[2]{\ensuremath{\langle {#1} , {#2} \rangle}}

\newcommand{\dom}[1]{\ensuremath{\mathrm{Dom} ({#1}) }}
\renewcommand{\dim}[1]{\ensuremath{\mathrm{dim} \left( {#1} \right) }}

\newcommand{\ran}[1]{\ensuremath{\mathrm{Ran} ({#1}) }}
\renewcommand{\ker}[1]{\ensuremath{\mathrm{Ker} ({#1}) }}

\newtheorem{thm}{Theorem}
\newtheorem{claim}{Claim}
\newtheorem{defn}{Definition}

\newtheorem{cor}{Corollary}

\begin{document}

\bibliographystyle{unsrt}

\title{Quantum Uncertainty and the Spectra of Symmetric Operators}

\author{R.T.W. Martin and A. Kempf}

\address{Departments of Applied Mathematics and Physics\\ University of Waterloo\\
Waterloo, Ontario N2L 3G1, Canada \\
phone: +1 519 888 4567 ext.35462 \\ fax: +1 519 746 4319}

\email{rtwmartin@gmail.com}

\begin{abstract}

In certain circumstances, the uncertainty, $\Delta S [\phi]$, of a
quantum observable, $S$, can be bounded from below by a finite
overall constant $\Delta S>0$, \emph{i.e.}, $\Delta S [\phi] \geq \Delta
S$, for all physical states $\phi$. For example, a finite lower
bound to the resolution of distances has been used to model a
natural ultraviolet cutoff at the Planck or string scale. In
general, the minimum uncertainty of an observable can depend on
the expectation value, $t=\langle \phi, S \phi\rangle$, through a
function $\Delta S_t$ of $t$, \emph{i.e.}, $\Delta S [\phi]\ge \Delta
S_t$, for all physical states $\phi$ with $\langle \phi, S
\phi\rangle=t$. An observable whose uncertainty is finitely
bounded from below is necessarily described by an operator that is
merely symmetric rather than self-adjoint on the physical domain.
Nevertheless, on larger domains, the operator possesses a family
of self-adjoint extensions. Here, we prove results on the
relationship between the spacing of the eigenvalues of these
self-adjoint extensions and the function $\Delta S_t$. We also
discuss potential applications in quantum and classical
information theory.

\vspace{5mm}   \noindent {\it Key words and phrases} :
self-adjoint extensions of symmetric operators, generalized
observables, finite minimum uncertainty, spectra of symmetric
operatotrs
\vspace{3mm}
(self-adjoint operator theory in quantum theory, including
spectral analysis), 47A10 (general theory of linear operators;
spectrum, resolvent)

\end{abstract}

\maketitle

\section{Introduction}

The uncertainty, $\Delta S[\phi]$, of a quantum observable, $S$,
can possess a finite lower bound $\Delta S>0$, \emph{i.e.}, $\Delta S
[\phi]\ge \Delta S$, for all physical states $\phi$. A simple
example is the momentum operator of a particle confined to a box
with Dirichlet boundary conditions. Since the position uncertainty
is bounded from above, the uncertainty relation implies that the
momentum uncertainty is finitely bounded from below. Another
example arises from general arguments of quantum gravity and
string theory, \cite{grossetal}, which point towards corrections
to the uncertainty relations which are of the type $\Delta x
~\Delta p \ge \frac{\hbar}{2}(1 + \beta (\Delta p)^2+...)$. For
positive $\beta$, this type of uncertainty relation implies a
finite lower bound to the position uncertainty. A Hilbert space
representation and functional analytic investigation of the
underlying type of generalized commutation relations first
appeared in the context of quantum group symmetric quantum
mechanics and quantum field theory, \cite{ak-earliest}, followed
by representations in quantum mechanics and quantum field theory
with undeformed symmetries, \cite{earliest-undeformed}. This made
it possible to implement this type of ultraviolet cutoff in
various quantum mechanical systems, see \emph{e.g.} \cite{brau}, as well
as in quantum field theory, with applications, \emph{e.g.}, in the study
of black hole radiation and inflationary cosmology, see \emph{e.g.}
\cite{bh-infl}.

Our aim in the present paper is to extend the basic functional
analytic understanding of observables whose uncertainty is
finitely bounded from below. We will consider the general case
where $\Delta S$ can be a function $\Delta S_t$ of the expectation
value, $t=\langle \phi, S \phi\rangle$, \emph{i.e.}, $\Delta S [\phi]\ge
\Delta S_t$, for all physical states $\phi$ with $t=\langle \phi,
S \phi\rangle$. As we will explain below, such observables are
necessarily described by operators that are merely symmetric
rather than self-adjoint on their domain in the Hilbert space.
Each such symmetric operator possesses, nevertheless, a family of
self-adjoint extensions to larger domains in the Hilbert space.
The spectra of these self-adjoint extensions are discrete. Our aim
here will be to prove results on the close relationship between
the spacing of the eigenvalues of these self-adjoint extensions
and the function $\Delta S_t$.

\section{Symmetric operators}

    Let $S$ be a closed, symmetric operator defined on a dense domain,
$\dom{S}$, in a separable Hilbert space $\mc{H}$. Recall that the
deficiency indices $(n_+, n_-)$ of $S$ are defined as the dimensions of
the subspaces $\ran{S + z} ^\perp= \ker{S^* -z} $ and $\ran{S-z} ^\perp = \ker{S^* +z}$
respectively where $z$ belongs to the open complex upper half plane (UHP).
The dimensions of these two subspaces are constant for $z$ within the
upper and lower half plane respectively (\cite{Glazman}, section 78). For
$z=i$ we will call $ \mc{D} _+ := \ker{S^* -i}$ and $\mc{D} _- := \ker{S^*
+i}$ the deficiency subspaces of $S$.

    We will let $\sigma (S)$, $\sigma _p (S) $, $\sigma _c (S)$,
$\sigma _r (S)$, and $\sigma _e (S)$ denote the spectrum, and the
point, continuous, residual and essential spectrum of $S$
respectively. Recall that $\sigma (S)$ is defined as the set of
all $\la \in \bm{C}$ such that $(S-\la)$ does not have a bounded
inverse defined on all of $\mc{H}$. The point spectrum $\sigma _p
(S)$ is defined as the set of all eigenvalues, $\sigma _c (S)$ is
here defined as the set of all $\la$ such that $\ran{S-\la}$ is
not closed, $\sigma _r (S)$ is defined as the set of all $\la $
such that $\la \notin \sigma _p (S)$ and $\ran{S-\la }$ is not
dense, and $\sigma _e (S)$ is the set of all $\la $ such that $S-
\la $ is not Fredholm. Recall that a closed, densely defined
operator $T$ is called Fredholm if $\ran{T}$ is closed and if the
dimension of $\ker{T}$ and the co-dimension of $\ran{T}$ are both
finite. If $T$ is unbounded, we include the point at infinity as
part of the essential spectrum. Clearly all the above sets are
subsets of $\sigma (S)$, and $\sigma (S) = \sigma _p (S) \cup
\sigma _c (S) \cup \sigma _r (S)$.

    If $S$ is symmetric, and $\la \in \bm{C} \setminus \bm{R}$,
then it is easy to see that $S-z$ is bounded below by
$\frac{1}{\mathrm{Im} (z)}$. This shows that any non-real $z \in \sigma
(S)$ must belong to the residual spectrum $\sigma _r (S)$ of $S$. If $S$
has finite deficiency indices, then the orthogonal complement of
$\ran{S-z}$ is finite dimensional for any $z \in \bm{C} \setminus \bm{R}$,
which shows that $\sigma _e (S) \subset \bm{R}$.

\noindent    The domain of the adjoint $S ^*$ of $S$ can be
decomposed as (\cite{Glazman}, pg. 98): \be \dom{S^*} = \dom{S}
\dotplus \mc{D} _+ \dotplus \mc{D} _- . \label{eq:neumann} \ee
Here the linear manifolds $\dom{S}$ , $\mc{D} _+$ and $\mc{D} _-$
are non-orthogonal, linearly independent, non-closed subspaces of
$\mc{H}$. The notation $\dotplus $ denotes the non-orthogonal
direct sum of these linear subspaces. If $S$ has finite deficiency
indices, and if the co-dimension of $\ran{S-\la}$ is infinite,
then $\la \in \bm{R}$. Furthermore if $\la \in \ran{S-\la} ^\perp$
then $\ov{\la}$ is an eigenvalue to $S^*$. This and the fact that
the dimension of $\dom{S^*}$ modulo $\dom{S}$ is finite (by the
above equation (\ref{eq:neumann})) allows one to conclude that $\la $
must be an eigenvalue of infinite multiplicity to $S$. Hence if
$\la \in \sigma _e (S)$ then either it is an eigenvalue of
infinite multiplicity or it belongs to the continuous spectrum of
$S$.

\begin{claim}
    If $S$ is a symmetric operator with finite and equal
deficiency indices then $\sigma _e (S) = \sigma _e (S')$,
and $\sigma _c (S) = \sigma _c (S')$ for any
self-adjoint extension $S'$ of $S$ within $\mc{H}$.
\label{claim:espec}
\end{claim}

    The domain of any self-adjoint extension $S'$ of $S$ can be
written as (\cite{Glazman}, pg. 100): \be \dom{S'} = \dom{S}
\dotplus (U-1) \mc{D} _+. \ee Here, $U$ is the isometry from
$\mc{D} _+ $ onto $\mc{D} _-$ that defines the self-adjoint
extension $S'$. Since the domain of $S$ and $S'$ differ by a
finite dimensional subspace, so do the range of $S'$ and $S$.
Using these facts it is straightforward to establish the claim.

\section{Minimum uncertainty and spectra of self-adjoint extensions}
\label{section:uncer}

As was first pointed out in \cite{ak-first-sampling}, there exists
a close relationship between the finite lower bound $\Delta S_t$
on the uncertainty of a symmetric operator and the spectra of its
self-adjoint extensions. Our aim now is to refine those results,
and to include new results, in particular, concerning the
dependence of the density of eigenvalues on the operator's
deficiency indices.

\begin{defn}
We denote the expectation value and the uncertainty of a symmetric
operator $S$ with respect to a unit-length vector $\phi \in
\dom{S}$ by $\bar{S} _\phi := \langle \phi, S \phi\rangle$ and by
$\Delta S[\phi ] := \sqrt{ \ip{S\phi}{S\phi} - (\ip{S\phi}{\phi} )
^2}$ respectively. For a fixed expectation value $t\in \bm{R}$,
the quantity $\Delta S_t := \inf _{\phi \in \dom{S} ,
\ip{S\phi}{\phi} =t , \| \phi \| =1 } \Delta S [\phi ] $ will be
called the minimum uncertainty of $S$ at $t$. The overall lower
bound on the uncertainty of $S$ will be denoted by $\Delta S :=
\inf_{t \in \bm{R} } \Delta S_t $.
\end{defn}

    Recall that a symmetric operator $S$ is said to be simple if there is
no subspace $\mf{S} \subset \mc{H}$ such that $S | _\mf{S}$ is
self-adjoint. A point $z \in \bm{C}$ is said to be a point of
regular type for $S$ if $(S-z)$ has a bounded inverse defined on
$\ran{S-z}$. As discussed in the introduction, every point $z \in
\bm{C} \setminus \bm{R}$ is a point of regular type for a
symmetric operator $S$. $S$ is then said to be regular, if every
$z \in \bm{C}$ is a point of regular type for $S$. It is clear
that if $\phi $ is an eigenvector of $S$, then $\Delta S[\phi ]
=0$. Furthermore, if $\la \in \bm{R} $ belongs to $\sigma _c (S)$ then $\Delta S_\la = 0$. Hence, if $\Delta S
\geq \eps
>0$ this implies that $S$ has no continuous or point spectrum on the real
line. This means, in particular, that such an $S$ is not
self-adjoint, and is both simple and regular. In addition, the
theorem below shows that if $\Delta S
>0$ then $S$ must have equal deficiency indices.

\begin{thm}
    If $S$ is a symmetric operator with unequal deficiency indices, then $\Delta S = 0$.
Furthermore, $\Delta S _t =0$ for all $t \in \bm{R}$.
\label{thm:uneq}
\end{thm}

    In the proof of this theorem, it will be convenient to use
the Cayley transform of the symmetric operator $S$. Given $\la$
in the upper half plane ($UHP$), consider $\ka _\la (z) := \frac{( z- \la )}{(z- \ov{\la})}$.
If $S$ is self-adjoint, then $\ka (S)$ is unitary by the
functional calculus. More generally if $S$ is symmetric, $\ka _\la
(S)$ is a partially defined transformation which is an isometry
from $\ran{S-\ov{\la}}$ onto $\ran{S-\la}$ (\cite{Glazman},
sections 67 and 79). For convenience we will take $\la=i$, and
write $\ka _i (z) =: \ka (z)$. The linear map $V:=\ka (S) : \mc{D}
_+ ^{\perp} \rightarrow \mc{D} _- ^{\perp}$ is called the Cayley
transform of $S$. One can further show that if $\ka ^{-1} (z) :=
i\frac{1+z}{1-z}$ and $V = \ka(S)$, then  $\ka ^{-1} (\ka (z)) = z $
and $S = \ka ^{-1} (V)$. Recall that all symmetric extensions of the symmetric operator
$S$ can be constructed by taking the inverse Cayley transforms of partial isometric
extensions of the Cayley transform $V= \ka (S)$. For example, if $S$ has
deficiency indices $(n, m)$, one can define an arbitrary partial isometry $W$ from
$\mc{D} _+ $ into $\mc{D} _-$, and the inverse Cayley transform $\ka ^{-1} (V') =: S'$
of $V' := V \oplus W$ on $\mc{H} = \ker{S^* -i} \oplus \mc{D} _+$ will
be a symmetric extension of $S$.

    The proof of the above theorem will also make use of the Wold
decomposition for isometries. Recall that the Wold decomposition theorem
states that any isometry on a Hilbert space $\mc{H}$
is isometrically isomorphic to an operator $U \oplus \left( \oplus _{\alpha \in \La } R \right) $
on $\mc{H} _0 \oplus \left( \oplus _{\alpha \in \La} l^2 (\bm{N} ) \right)$ where $U$ is
some unitary on $\mc{H} _0$, $\La$ is some index set, and $R$
is the right shift operator on the Hilbert space of square summable
sequences, $l^2 (\bm{N})$ (see \emph{e.g} \cite{shift}, pg. 2).

\begin{proof}
    Suppose $S$ has deficiency indices $(n_+,n_-)$, $n_+ \neq n_-$, $n_\pm =  \dim{\mc{D} _\pm}$. Then $S$ has a
symmetric extension $S'$ with deficiency indices either
$(0,j)$ or $(j, 0)$, where $j:=| n_+ -n_- |$. Recall that such an extension
is obtained as follows. Take the Cayley transform $V$ of $S$, and in the case
where $n_-> n_+$, extend it by an arbitrary isometry from $\mc{D} _+$ into $\mc{D} _-$ to
obtain an isometry $V'$ with $\dim{\ran{V'} ^\perp} =j$. The inverse Cayley transform
of $V'$ yields the desired symmetric extension $S'$ with deficiency indices $(0,j)$. In the case
where $n_+ > n_-$, extend $V$ by an arbitrary isometry from an arbitrary $n_-$ dimensional
subspace of $\mc{D} _+$ onto $\mc{D} _-$ to obtain a partial isometry $V'$ with
$\dim{\ker{V'}} =j$ and $\dim{\ran{V'} ^\perp} = 0$. Again, the inverse Cayley transform
of $V'$ yields the desired symmetric extension $S'$ for this case with deficiency indices $(j,0)$.

Accordingly, the Cayley transform $V'$ of $S'$ is either an isometry with\\
$\dim{\ran{V'} ^\perp}=j$ or the adjoint of an isometry,
with $\dim{\ker{V'}} =j$. By the Wold decomposition theorem, $V'$ is isometrically
isomorphic to the direct sum of a unitary operator $U$
and either $j$ copies of the right shift operator or $j$ copies of the left shift
operator on $\mc{H} _0 \oplus _{i=1} ^{j} l^2 (\mathbb{N} )$. It follows
that $\sigma (V') \supset \sigma (R)$ or $\supset \sigma (L)$ respectively,
where $R$ and $L$ are the right and left shift operators on $l^2
(\bm{N})$. It is known that the unit circle lies in the continuous
spectrum of both the right and left shift operators. It is not
difficult to see that $\la \in \sigma _c (V') \setminus \{ 1 \}$ where $V'$
is the Cayley transform of $S'$ if and only if $\ka ^{-1} (\la)
\in \sigma _c (S') = \sigma _c (S)$.  It follows that the
continuous spectrum of $S$ (which is a subset of $\bm{R}$) is
non-empty and hence there exist $\phi \in \dom{S}$ for which
$\Delta S [\phi ]$ is arbitrarily small.

Furthermore, the above shows that the continuous
spectrum of $S$ is all of $\bm{R}$. Using this fact, it is not difficult to show that
$\Delta S _t =0$ for all $t\in \bm{R}$. First, given any fixed $t \in \bm{R}$, since
$t \in \sigma _c (S)$, one can find a sequence $(\phi _n )_{n \in \bm{N}} \subset \dom{S}$
such that
$\| (S-t) \phi_n \| \rightarrow 0$, $\| \phi _n \| =1 $, and $t_n := \ip{S\phi_n}{\phi _n}
\geq t$. Again, since $\sigma _c (S) = \bm{R}$, one can find a unit norm $\psi \in \dom{S}$
such that $\ov{S} _\psi = t' < t$. Let $P_n $ denote the projectors onto the orthogonal complements
of the one dimensional subspaces spanned by the $\phi _n$. Then each vector $P_n \psi \in
\dom{S}$, and it is easy to verify that if $\psi _n := \frac{P_n \psi}{\| P_n \psi \|}$,
then $\ov{S} _{\psi _n  } =: t' _n \leq  t $. For each $n \in \bm{N}$, let $\varphi _n$ be the linear
combination of the  $\psi _n $ and $\phi _n$ such that $\| \varphi _n \| =1 $ and $ \ov{S} _{\varphi _n} = t$.
It is straightforward to verify that $(S-t) \varphi _n \rightarrow 0$ so that $\Delta S _t =0$.
\end{proof}

    Note that the above theorem implies that if $S$ is
any symmetric operator with unequal deficiency indices that represents
a quantum mechanical observable, then even though $S$ is not self-adjoint,
it is possible to measure that observable as
precisely as one likes in the sense that $\Delta S _t = 0 $ for all $t \in \bm{R}$. Nevertheless,
despite the fact that $\Delta S _t =0$ for all $t  \in \bm{R}$, the situation is physically different from the
case of a self-adjoint observable. This is because, in this case, the formal, non-normalizable quasi-eigenstates
of the symmetric operator $S$ are non-orthogonal. For example, consider the case of the
symmetric derivative operator $D:= i\frac{d}{dx}$ defined with domain
$\dom{D}$ in $L^2 [0 ,\infty)$, $\dom{D} := \{ \phi \in L^2 [0 ,\infty) \ | \ \phi \in AC_{loc} [0,\infty) ; \ D \phi
\in L^2 [0 ,\infty); \ \phi (0) =0 \}$. Here, $AC _{loc} [0, \infty)$ denotes the set of all functions which are
absolutely continuous on any compact subinterval of $[0 , \infty)$. It is straightforward to check that $D$ has deficiency
indices $(0,1)$.  If $\phi _\la (x) := \frac{e^{-i\la x}}{\sqrt{2\pi}}$, for $x \in [0, \infty)$,
then $\phi _\la $ can be thought of as a formal, non-normalizable quasi-eigenstate for $S$, since
if $f \in L^2 [0, \infty)$, the formal inner product of $f$ with the $\phi _\la $,
\be \int _0 ^\infty f (x) \frac{1}{\sqrt{2\pi}} e^{i\la x} =: F(\la ), \ee generates a unitary transformation,
\emph{i.e.} the Fourier transform, of $L^2 [0, \infty)$ onto a subspace of $L^2 (\bm{R} )$, and under this transformation
$S$ is transformed into multiplication by the independent variable. These quasi-eigenstates are non-orthogonal in
the following sense. For $\eps >0 $ and $\la \in \bm{R}$, let
$\phi (\eps , \la ; x) := \frac{1}{\sqrt{2\pi}} \left( e^{-i\la x - \eps x } - e^{-\frac{x}{\sqrt{\eps}} } \right)$. Then
$\phi (\eps , \la ; x ) \in \dom{D}$ for any $\eps >0$, and as $\eps \rightarrow 0$, $\phi (\eps , \la ; x)$ converges
to $\phi _\la$ in $L^2 $ norm on any compact subinterval of $[0, \infty)$. Furthermore, it is straightforward to
check that $\frac{ \ip{D\phi (\eps ,\la ; \cdot) }{\phi (\eps , \la ; \cdot)}}{\| \phi (\eps , \la ; \cdot ) \| ^2 }
\rightarrow \la $ and that $\frac{ \ip{D\phi (\eps ,\la ; \cdot) }{D \phi (\eps , \la ; \cdot)}}{\| \phi (\eps , \la ; \cdot ) \| ^2 }
\rightarrow \la ^2 $ as $\eps \rightarrow 0$. However, if $\la _1 \neq \la _2 \in \bm{R}$, then the inner
product $\ip {\phi (\eps , \la _1 ; \cdot}{\phi (\eps , \la _2 ; \cdot )} $ converges to $\frac{1}{2\pi i (\la _2 -\la _1)} \neq 0$
in the limit as $\eps \rightarrow 0$. In this sense, the formal non-normalizable quasi-eigenstates  $\phi _\la $
are not orthogonal. Compare this to the case of the self-adjoint derivative operator $D' := i \frac{d}{dx}$ in $L^2 (\bm{R})$.
In this case the non-normalizable eigenstates to eigenvalues $\la \in \bm{R}$ are again $\phi _\la (x) = \frac{1}{\sqrt{2\pi}}
e^{-i\la x}$, $x \in \bm{R}$. If one chooses, for example, $\phi (\eps, \la ; x ) := \frac{1}{\sqrt{2\pi} } e^{-i\la x - \eps x^2}
\in \dom{D'}$, then it is straightforward to check that the inner product $\ip{\phi (\eps , \la _1 ; \cdot)}{\phi(\eps ,
\la _2 ; \cdot ) } $ vanishes as $\eps \rightarrow 0$ if $\la _1 \neq \la _2$, so that the non-normalizable eigenstates of
this self-adjoint operator are indeed orthogonal.

    For a concrete physical example, consider a telescope with some finite aperture. The accurate measurement of the
arriving photons' momentum orthogonal to the telescope is essential for the production of a sharp image. This amounts to
the measurement of the momentum of a particle in a box. The momentum operator (which is just $i$ times the first
derivative operator) acting on a particle in a box is a symmetric operator with deficiency indices (1,1). In this case
the finite aperture of the telescope causes a minimum uncertainty in the angle measurements \cite{ak-first-sampling}.
The case of a telescope is the
case of light being diffracted as it passes through a slit. The case of the symmetric derivative operator $D$ on the
half line, $L^2 [0,  \infty)$, is that of light being diffracted at a single edge, \emph{i.e.},
passing a single wall. The fact that the quasi eigenstates are not orthogonal means, physically,
 that there is a diffraction pattern in this case as well. \vspace{.1in}

     If $S$ is a simple symmetric operator with deficiency indices
$(j,0)$, or $(0,j)$, then it is straightforward to verify that $S$ is
isometrically isomorphic to $j$ copies of the differentiation operator
$i\frac{d}{dx}$ on $L^2 (0, \infty )$ or $L^2 (-\infty ,0)$ respectively. This fact was
first proven by von Neumann, see for example, (\cite{Glazman}, Section 82). Hence, if $S$ has
deficiency indices $(j,0)$, or $(0,j)$, it generates a semi-group
of isometries or co-isometries which is isometrically isomorphic to $j$ copies of
right translation on $L^2 (0, \infty)$ or $L^2 (-\infty , 0)$, respectively. It
follows that if $S$ is a symmetric quantum mechanical Hamiltonian operator, which has
deficiency indices $(m,n)$, and $j:=|n-m|$, then any maximal symmetric extension of $S$
will generate either an isometric or co-isometric time evolution of the quantum mechanical
system. Furthermore, the Hilbert space can be decomposed into $j+1$ subspaces such that the time-evolution
on the first subspace is unitary, and such that the time evolution on each of the other subspaces
is purely isometric or co-isometric. If the state of the system begins in one of these subspaces,
its image at any later time will be confined to that subspace, so that there are, in general,
subspaces of the Hilbert space which will be inaccessible to the time evolution of the system
once the initial state is fixed.

\begin{thm}
Let $S$ be a densely defined, closed symmetric operator with
finite and equal deficiency indices $(n,n)$. If $\Delta S > 0$,
then any self-adjoint extension $S'$ of $S$ has a purely discrete
spectrum, $\sigma (S') = \sigma _p (S')$. In particular, if
$\Delta S _t > \eps >0$ for all $t \in I \subset \bm{R}$, then
$S'$ can have no more than $n$ eigenvalues (including
multiplicities) in any interval $J \subset I$ of length less than
or equal to $\eps$, and if $n=1$, then $S'$ can have no more than one
eigenvalue in any interval $J \subset I$ of length less than or equal
to $2\eps$.  \label{thm:minun}
\end{thm}

    This theorem shows, in particular, that if $\Delta S > \eps$,
then any self-adjoint extension of $S$ has no more than $n$
eigenvalues in any interval of length $\eps$. The authors are currently investigating whether
the improved result that holds for the $n=1$ case
generalizes to the case of higher deficiency indices.

\begin{proof}

    If $\Delta S >0$, then as in the discussion preceding the
proof of Theorem \ref{thm:uneq}, we conclude that every $z \in
\bm{C}$ is a point of regular type for $S$. Since $S$ has finite
and equal deficiency indices, if $S'$ is any self-adjoint
extension of $S$, it follows that $\sigma _e (S') = \sigma _e (S)$
consists only of the point at infinity. This implies that the
spectrum of $S'$ consists soley of eigenvalues of finite
multiplicity with no finite accumulation point.

    Suppose that there is a self-adjoint extension $S'$ of $S$ which
has $n+1$ eigenvectors $\phi _i$ to eigenvalues $\lambda _i$ where
$\la _i \in J \subset I$, and the length of $J$ is less than or
equal to $\eps$. Then since the dimension of $\dom{S'} $ modulo
$\dom{S}$ is $n$, there is a non-zero linear combination of these
orthogonal eigenvectors, $\psi = \sum _{i=1} ^{n+1} c_i \phi _i$
which has unit norm and which belongs to $\dom{S}$. The
expectation value of the symmetric operator $S$ in the state $\psi
$ lies in $J$, $t:= \ov{S} _\psi \in J$ since $\psi$ is a linear
combination of eigenvectors to $S '$ whose eigenvalues all lie in
$J$. Now it is straightforward to verify that since $|\lambda _i |
< |t| + \eps  $ for all $1 \leq i \leq n+1$, that \be \left(
\Delta S [\psi ] \right) ^2  = \sum _{i=1} ^{n+1} \lambda _i ^2
|c_i |^2 - t ^2 \label{eq:bound} \leq \sum _{i=1} ^{n+1} (|t| +
\eps ) ^2 |c_i |^2  - t ^2 = 2 |t | \eps  + \eps  ^2 \ee

    Now first suppose that $0 \in J$ and that
$t := \ov{S} _\psi =0$. Then in this case equation
(\ref{eq:bound}) contradicts the fact that $\Delta S _0 > \eps $,
proving the claim for this case.

    If $t \neq 0$, then consider the symmetric operator $S (t) := S
-t$ on $\dom{S}$. Given any $\phi \in \dom{S}$ which has unit norm
and expectation value $\ov{S} _\phi = \ip{S \phi}{\phi} =t$, it is
not hard to see that $\ov{S(t)} _\phi = \ip{S(t) \phi}{\phi} =0$
and that \ba \Delta S (t) [\phi ] &  = &  \ip{S(t) \phi}{ S(t)
\phi} = \ip{S \phi}{ S \phi} - 2t \ip{S\phi}{\phi} + t^2 \nonumber
\\ & = & \ip{S \phi}{ S \phi} - t^2 = \Delta S [\phi ] .\ea This
shows that $\Delta S(t) _0 = \Delta S _t > \eps $. Now let $S'$ be
any self-adjoint extension of $S$. Applying the above result for
the expectation value $0$ to the symmetric operator $S(t)$, we
conclude that the self-adjoint extension $S' (t) := S' -t $ of
$S(t)$ can have no more than $n$ eigenvalues in the interval $J -
t$. This in turn implies that $S'$ can have no more than $n$
eigenvalues in the interval $J$.

    Now suppose that $n=1$. As in the above, if $S'$ is a self-adjoint extension of
$S$ that has two eigenvectors $\phi _1, \phi _2$ to eigenvalues $\la_1 ,\la _2$,
in a subinterval $J\subset I$ of length less than or equal to $2\eps$, then there is
a unit norm vector $\psi = c_1 \phi _1 + c_2 \phi _2$, that belongs to $\dom{S}$.
The expectation value of $\psi$, $t:= \ip{S\psi}{\psi}$ will also belong to $J$.
The expectation value $t$ and the fact that $\psi $ has unit norm, uniquely determines
the constants $c_1$ and $c_2$ up to complex numbers of modulus one:
\be |c_1| = \sqrt{\frac{|\la _2 -t|}{| \la _1 -\la _2 | } } \ \ \mathrm{and}~~
|c_2 | = \sqrt{ \frac{ | \la_1  -t | }{ |\la _1 - \la _2 | }}. \ee It is now straightforward
to calculate that $\Delta S [\psi ] = \sqrt{|\la _1 - t| |\la _2 -t|}$. Assume, without
loss of generality, that $\la _1 < \la _2$, so that $\la _1 < t < \la _2$. Since $\la _1 , \la _2$
belong to the same interval $J$ with length less than or equal to $2 \eps $, it follows that
$| \la _2 -t | \leq 2\eps - |\la _1 -t |$, so that $\Delta S [\psi ] \leq \sqrt{ (2\eps |\la _1 -t| -|\la _1 -t| ^2) }$.
It is simple to check that the function $f(x) = 2\eps x - x^2 $ has a global maximum
of $\eps ^2$ when $x = \eps $, so that $\Delta S [\psi ] \leq \eps$. Since $t \in I$, this
contradicts the assumption that $\Delta S _t > \eps$.

\end{proof}

\begin{cor}
    If $S$ is a symmetric operator with finite deficiency indices
such that $\Delta S = \eps >0$, then $S$ is simple, regular, the
deficiency indices $(n,n)$ of $S$ are equal, and the spectrum of any
self-adjoint extension of $S$ is purely discrete and consists of
eigenvalues of finite multiplicity at most $n$ with no finite accumulation point.
\label{cor:sreg}
\end{cor}

    It is known (\cite{Glazman}, section 83) that if $S$ is a
closed, densely defined simple symmetric operator with equal and
finite deficiency indices $(n,n)$, then the multiplicity of any
eigenvalue of any self-adjoint extension $S'$ of $S$ does not
exceed $n$. Corollary \ref{cor:sreg} is an immediate consequence
of this fact and Theorems \ref{thm:uneq} and \ref{thm:minun}.

    For example, consider the symmetric differential
operator $S':=-\frac{d}{dx} \left( x \frac{d}{dx} \cdot \right) +
x $ defined on the dense domain $C_0 ^\infty (0,\infty ) \subset
L^2 [0 , \infty)$ of infinitely differentiable functions with
compact support in $(0, \infty)$. Let $S$ be the closure of $S'$.
Let $D$ be the closed symmetric derivative operator on $L^2 [0,
\infty)$ which is the closure of the symmetric derivative operator
$D' := i \frac{d}{dx}$ on the domain $\mf{D} := C_0 ^\infty (0 ,
\infty )$. It follows that $\mf{D}$ is a core for both $D$ and for
$S$. Recall that a dense set of vectors $\mc{D}$ is called a core
for a closable operator $T$, if $\ov{ T | _{\mc{D}}} = \ov{T}$.
For all $\phi \in \mf{D}$, it is easy to verify that $[D , S]
\phi := (DS - SD) \phi = i (D^2 +1) \phi$. By the uncertainty
principle, it follows that for any $\phi \in \mf{D}$, \ba \Delta S
[\phi ] \Delta D [\phi ] & \geq & \frac{1}{2} | \ip{\phi}{[D ,S] \phi}
| = \frac{1}{2} \ip{\phi}{(D^2 +1)\phi} \nonumber \\ & = & \frac{1}{2} \left( \Delta
D [\phi ]  ^2 + \ip{\phi}{D\phi} ^2 + \ip{\phi}{\phi} \right).\ea
Using the fact that the function $f(t) = \frac{t^2 +1}{2t} $ is
concave up for all $t \in (0 ,\infty)$ and has a global minimum
$f(1) =1$, we conclude that $\Delta S [\phi ] \geq 1 $ for any
$\phi \in \mf{D}$. Since $\mf{D}$ is a core for $S$, given any
$\psi \in \dom{S}$ we can find a sequence $\psi _n \in \mf{D}$
such that $\psi _n \rightarrow \psi$ and $S \psi _n  \rightarrow S
\psi$. It follows that $\Delta S [\psi ] = \lim _{n \rightarrow
\infty} \Delta S [ \psi _n ] \geq 1$. This shows that $\Delta S
\geq 1$. Now $S$ is a second order symmetric differential
operator. It is known that the deficiency indices of such an
operator are equal and do not exceed $(2,2)$ (\cite{Naimark},
Section 17). Since $\Delta S \geq 1$, Corollary \ref{cor:sreg}
also implies that the deficiency indices of $S$ must be equal and
non-zero. Hence $D$ has deficiency indices $(1,1)$ or $(2,2)$.
Applying Theorem \ref{thm:minun} one can now conclude that any
self-adjoint extension of $S$ can have at most two eigenvalues in
any interval of length $1$.

   \noindent Conversely, if $S$ has finite deficiency indices and is simple
and regular, then $\Delta S >0$.

\begin{thm}
 Suppose that $S$ is a regular, simple symmetric operator with
 finite and equal deficiency indices. Let $\mf{S}$ denote the set
 of all self-adjoint extensions of $S$ within $\mc{H}$. Then, \be \Delta S _t \geq \max
 _{S' \in \mf{S}} \Delta S' _t = \max _{S' \in \mf{S}} \left( \min
 _{\la , \mu \in \sigma (S')} \sqrt{|\la - t||\mu -t | } \right).
 \label{eq:minunt} \ee \label{thm:sreg}
\end{thm}
\begin{proof}
    Note that if $S$ is simple and regular with finite deficiency
indices $(m,n)$, then these indices must be equal, otherwise $S$
would have continuous spectra and would not be regular.

    Since $\dom{S} \subset \dom{S'}$ and $S' | _{\dom{S}} = S$
for any $S' \in \mf{S}$, it is clear that $\Delta S _t \geq \max
_{S' \in \mf{S} } \Delta S'_t$. It remains to prove that $\Delta
S'_t  = \min _{\la , \mu \in \sigma (S') } \sqrt{ | \la - t | |
\mu - t | } $ for any $S' \in \mf{S}$. Since we assume $S$ is
regular, simple, and has finite deficiency indices, the essential
spectrum of $S$ is empty. Hence by Claim \ref{claim:espec},
$\sigma _e (S' ) $ is empty for any $S' \in \mf{S}$. This shows
that the spectrum of any $S'$ consists solely of eigenvalues of
finite multiplicity with no finite accumulation point. Order the
eigenvalues as a non-decreasing sequence $(t _n ) _{n \in \bm{M}}
 $ and let $\{ b_n \} _{n \in \bm{M}}$ be the corresponding
orthonormal eigenbasis such that $S' b_n = t _n b_n$. Here $
\bm{M} = \pm \bm{N}$ or $=\bm{Z}$, depending on whether $S'$ is
bounded above, below, or neither bounded above nor below.  To calculate
$\Delta S'_t$, we need to minimize the functional \be \Phi ' [
\phi ] := \ip{ S' \phi }{S' \phi} - t^2 \ee over the set of all
unit norm $\phi \in \dom{S'}$ which satisfy $\ip{S' \phi}{\phi}
=t$. Let us assume that $t$ is not an eigenvalue of $S'$ as in
this case $\Delta S'_t =0$ and (\ref{eq:minunt}) holds trivially.
Expanding $\phi$ in the basis $b_n$, $\phi = \sum _{n \in \bm{M}}
\phi _n b_n$, we see that to find the extreme points of $\Phi ' $
subject to these constraints we need to set the functional
derivative of \be \Phi [ \phi ] := \sum _{n \in \bm{Z}} \phi _n
\ov{\phi} _n \left( \phi _n \ov{\phi} _n (t_n ^2 -t ^2) - \alpha
_1 \phi _n \ov{\phi} _n t_n - \alpha _2 \right) \ee to zero. Here
$\alpha _1 ,\alpha _2 $ are Lagrange multipliers. Setting the
functional derivative of $\Phi$ with respect to $\ov{\phi}$ to $0$
yields: \be 0 = \phi _n \left( (t_n ^2 -t^2 ) - \alpha _1  t_n
-\alpha _2 \right) \label{eq:constraint} \ee Formula
(\ref{eq:constraint}) leads to the conclusion that if $\phi $ is
an extreme point, it must be a linear combination of two
eigenvectors to $S'$ corresponding to two distinct eigenvalues.

    To see this note that if the decomposition of $\phi$ in
the eigenbasis $\{ b_n \}$ had three non-zero coefficients, say
$\phi _{j_i}$, $i=1,2,3$, all of which correspond to eigenvectors
$b_{j_i}$ with distinct eigenvalues, $t_i \neq t_j$, $1 \leq i,j
\leq 3$, then Equation (\ref{eq:constraint}) leads to the conclusion
that $\alpha _1 = t_{j_1} + t_{j_2} = t_{j_2} + t_{j_3} $ which would imply that
$t_{j_1} = t_{j_3}$, a contradiction. Furthermore, $\phi$ cannot
just be a linear combination of eigenvectors $b_j$ to one
eigenvalue, as such a linear combination cannot satisfy the
constraint $\ip{S \phi }{\phi } =t$. So let $\la := t_i $ and $\mu
:= t_j $ for any $j, i \in \bm{Z}$ for which $t_i \neq t_j$.
Choose $\varphi \in \ker{S^* -\la} $ and $\psi \in \ker{S^* -\mu}$. We
have shown that $\phi$ has the form $\phi = c_1 \varphi + c_2
\psi$. As in the proof of Theorem \ref{thm:minun},
the constraints that $\ip{\phi }{\phi } = 1 $ and
$\ip{S \phi }{\phi }=t$ uniquely determine $c_1$ and $c_2$ up to
complex numbers of modulus one:

\be |c_1| = \sqrt{\frac{|\mu -t|}{| \la -\mu| } } \ \ \mathrm{and}~~
|c_2 | = \sqrt{ \frac{ | \la -t | }{ |\la - \mu | }}. \ee

    The phases of $c_1 $ and $c_2 $ do not affect the value of
$\Phi [ \phi ]$. It follows that if $\phi $ extremizes $\Phi$,
then $\Delta S' [\phi ] = \sqrt{ | \mu - t | | \la - t |} $ so
that $\Delta S'_t = \min _{\mu, \la \in \sigma (S' ) } \sqrt{ |
\mu - t | | \la - t |}$.
\end{proof}

    Observe that the curve $f(t) = \sqrt{ | \mu - t | | \la - t
|}$ describes the upper half of a circle of radius $\frac{|\la -
\mu| }{2} $ centred at the point $\frac{\la +\mu}{2}$.

\begin{cor}
    If $S$ is a symmetric operator with deficiency indices $(n,n)$
such that $\Delta S>0$, then $\max _{S' \in \mf{S}} \Delta S' _t \geq \frac{\Delta S}{\sqrt{2}}$.
If $n=1$, then $\max _{S' \in \mf{S}} \Delta S' _t \geq \Delta S$, so that
\be \Delta S = \inf _{t \in \bm{R}} \max
 _{S' \in \mf{S}} \Delta S' _t = \inf _{t \in \bm{R}} \max _{S' \in \mf{S}} \left( \min
 _{\la , \mu \in \sigma (S')} \sqrt{|\la - t||\mu -t | } \right). \label{eq:better} \ee

\end{cor}

\begin{proof}
    It is known that if
$\la $ is a regular point of a symmetric operator $S$ with deficiency indices $(n,n)$,
then there exists a self-adjoint extension of $S$ for which $\la $ is an eigenvalue of
multiplicity $n$ (\cite{Glazman}, pg. 109).  Given any $\la $
for which $|\la -t | < \Delta S $, let $S'$ be the self-adjoint extension of $S$ for
which $\la $ is an eigenvalue of multiplicity $n$. By Theorem \ref{thm:minun}, if $\mu \neq \la $
belongs to $\sigma (S')$, it must be that $|\mu -t | \geq \Delta S - |\la -t | $. Again, by
formula (\ref{eq:minunt}),  $\Delta S' _t \geq  \sqrt{|\la -t | \Delta S - (\la -t)^2 } $. It is
a simple calculus exercise to show that this is maximized when $|\la -t| = \frac{\Delta S}{2}$.
Choosing $\la$ so that this condition is satisfied proves the first part of the claim.

    Using the result of  Theorem \ref{thm:minun}
in the case where $n=1$, and repeating the above arguments shows that, in the case where $n=1$,
 $\max _{S' \in \mf{S}} \Delta S' _t \geq \Delta S$. By Theorem \ref{thm:sreg},
$\Delta S_t \geq \max _{S' \in \mf{S}} \Delta S' _t$, so that $\Delta S_t \geq \max _{S' \in \mf{S}}
\Delta S' _t \geq \Delta S$. Taking the infimum over $t\in \bm{R}$ of both sides yields
$\Delta S = \inf _{t\in \bm{R} } \max _{S' \in \mf{S}} \Delta S' _t$. Combining this with formula
(\ref{eq:minunt}) now yields the formula (\ref{eq:better}).
\end{proof}

    If the improved result of Theorem \ref{thm:minun} that holds for the $n=1$ case
could be established for all values of $n$,
then the stronger result of the above theorem for the $n=1$ case would also hold for all $n$.

\section{Outlook}

Our new results on operators whose uncertainty is bounded from below are
of potential interest in quantum gravity. This is because our results
improve on the results of \cite{ak-first-sampling}, where it was first
pointed out that physical fields in theories with a finite lower bound on
spatial resolution, $\Delta x$, possess the so-called sampling property: a
field is fully determined by its amplitude samples taken at the
eigenvalues of any one of the self-adjoint extensions of $x$. With this
ultraviolet cutoff, physical theories can therefore be written,
equivalently, as living on a continuous space, or as living on any one of
a family of discrete lattices of points. This provides a new approach to
reconciling general relativity's requirement that spacetime be a
continuous manifold with the fact that quantum field theories tend to be
well-defined only on lattices. This approach has been extended to quantum
field theory on flat and curved space see \cite{ak-flatandcurved}. Representing
Hamiltonians as symmetric operators with unequal deficiency indices may
also be of physical significance, since the co-isometric time evolution of a quantum
system generated by such a Hamiltonian could be useful for describing information
vanishing beyond horizons, \emph{e.g} the horizon of a black hole \cite{ak-shortdist}.

Finally, we remark that our results are also of potential interest in
information theory: The theory of spaces of functions which are determined
by their amplitudes on discrete points of sufficient density is a
long-established field, called sampling theory. Sampling theory plays a
central role in information theory, where it serves as the crucial link
between discrete and continuous representations of information, see
\cite{shannon}. Our results on the relationship between the varying
uncertainty bound $\Delta S_t$ and varying density of the eigenvalues of
the self-adjoint extensions of $S$ therefore contribute new tools for
handling the difficult case of sampling and reconstruction with a variable
Nyquist rate.

\pagebreak


\begin{thebibliography}{00}
\bibitem{ak-shortdist} A. Kempf, \it On the only three Short Distance Structures which can be described by Linear Operators,
 \rm Rept.Math.Phys. {\bf 43}, 171-177 (1999), hep-th/9806013.
\end{thebibliography}
\end{document}